\documentclass[conference]{llncs}
\usepackage[cmex10]{amsmath}
\usepackage{amssymb}

\newcommand{\pnt}[1]{{\mbox{\boldmath $#1$}}}
\newcommand{\cof}[2]{\mbox{$#1_{\boldsymbol{#2}}$}}

\newcommand{\V}[1]{\mbox{$\mathit{Vars}(#1)$}}
\newcommand{\Va}[1]{\mbox{$\mathit{Vars}(\boldsymbol{#1})$}}

\newcommand{\s}[1]{\mbox{$\{#1\}$}}

\newcommand{\nGz}[2]{$G_{non-\{z\}}$}

\newcommand{\prr}[1]{\mi{Prev}(\boldsymbol{q})}

\newcommand{\mi}[1]{\mathit{#1}}
\newcommand{\ti}[1]{\textit{#1}}
\newcommand{\tb}[1]{\textbf{#1}}

\newcommand{\Dds}[2]{\mbox{\pnt{#1}~$\rightarrow #2$}~}
\newcommand{\Dss}[4]{\mbox{$(\prob{#1}{#2},\pnt{#3})~\rightarrow #4$}}

\newcommand{\ttt}{\>\>\>}

\newcommand{\Tt}{\>\>}

\newcommand{\Sup}[2]{\mbox{$#1^\mi{#2}$}}

\newcommand{\prob}[2]{\mbox{$\exists{#1} [#2]$}}
\newcommand{\DS}{\mbox{$\Omega$}}

\newcommand{\slv}{DS-QSAT}
\newcommand{\cl}{SAT\_ALG}
\newcommand{\Slv}{\ti{DS-QSAT}}

\usepackage{wrapfig}
\usepackage{graphicx}

\begin{document}

\title{Checking Satisfiability  by Dependency Sequents}

\author{Eugene Goldberg,  Panagiotis Manolios}
\institute{Northeastern University, USA \email{\{eigold,pete\}@ccs.neu.edu}}

\maketitle

\begin{abstract}
We introduce a new algorithm for checking satisfiability based on a calculus of 
Dependency sequents (D-sequents).
Given a CNF formula $F(X)$,  a D-sequent is a record stating that under a partial assignment a set
of variables of $X$ is redundant in formula \prob{X}{F}. The D-sequent calculus
is based on operation \ti{join} that  forms a new D-sequent from
two existing ones. The new algorithm solves the \ti{quantified} version of SAT.
That is, given a satisfiable formula $F$, it, in general, does not produce an assignment satisfying $F$.
 The new algorithm is called
\slv~where DS stands for Dependency Sequent and Q for Quantified.
 Importantly, a DPLL-like procedure is only a special case
of \slv~where  a very restricted kind of D-sequents is used. We argue
that this restriction a) adversely affects scalability of SAT-solvers
b) is caused by looking for an \ti{explicit} satisfying
assignment rather than just proving satisfiability. We give experimental results
substantiating these claims.

\end{abstract}

\section{Introduction}
Algorithms for solving the Boolean satisfiability problem  are  an important part of  modern design flows.
Despite  great progress in the performance of such algorithms achieved recently,
the scalability of SAT-solvers still remains a big issue.
In this paper, we address this issue  by introducing  a new method of satisfiability checking that can be viewed  as a descendant of the
DP procedure~\cite{dp}.

We consider Boolean formulas represented in Conjunctive Normal Form (CNF). Given a CNF formula $F(X)$,  
one can formulate two different kinds of satisfiability checking problems. We will refer to the problems of the first kind as \tb{QSAT} where
Q stands for quantified. Solving QSAT means just checking if \prob{X}{F} is true. In particular, if $F$ is satisfiable, a QSAT-solver
does not have to produce an assignment satisfying $F$. The problems of the second kind that we will refer to as just \tb{SAT} are a special
case of those of the first kind. If $F$ is satisfiable, a SAT-solver has to produce an assignment satisfying $F$.

Intuitively, QSAT should be easier than SAT because a QSAT-solver needs to return only one bit of information. This intuition is substantiated
by the fact that checking if an integer number $N$ is prime (i.e. answering the question if non-trivial factors of $N$ \ti{exist})
 is polynomial while finding factors of  $N$ \ti{explicitly}  is believed to be hard. 
However, the situation among practical  algorithms defies this intuition. Currently, the field is dominated by  procedures based on DPLL algorithm
~\cite{dpll} that is by SAT-solvers. On the other hand,  a classical  QSAT-solver, the DP procedure~\cite{dp}, does not have any competitive descendants
(although some elements of the DP procedure are used in formula preprocessing performed by SAT-solvers ~\cite{prepr}). 

In this paper, we introduce a QSAT-solver called \tb{DS-QSAT} where DS stands for 
Dependency Sequent. On the one hand, \slv~can be viewed as a descendant of the DP procedure. On the other hand, DPLL-like procedures
with clause learning is a special case of \slv. Like DP procedure, \slv~is based on the idea of elimination of redundant variables. 
A variable $v \in X$ is redundant in
\prob{X}{F} if the latter is equivalent to \prob{X}{F \setminus F^v} where $F^v$ is the set of all clauses of $F$ with $v$. Note that removal
of clauses of $F^v$ produces a formula that is  equisatisfiable rather than  functionally equivalent to $F$.

If $F$ is satisfiable,  all variables of $X$ are redundant in \prob{X}{F} because an empty set of clauses is satisfiable.
If $F$ is unsatisfiable, one can make the variables of $F$ redundant by deriving an empty clause and adding it to  $F$.
An empty clause  is unsatisfiable, hence  all other clauses of $F$ can be dropped.
So, from the viewpoint of \slv, the only difference between  satisfiable and unsatisfiable formulas  is as follows.
If $F$ is satisfiable, its variables  are \ti{already} redundant and one just needs to prove this redundancy. If $F$ is unsatisfiable,
one has to \ti{make} variables redundant by derivation and adding to $F$ an empty clause.

The DP procedure makes a variable $v$ of $X$ redundant \ti{in one step}, by adding to $F$ all  clauses that can be produced by resolution on  $v$.
This is extremely inefficient due to generation of prohibitively large sets of clauses  even for very small formulas. \slv~addresses this
problem by using branching. The idea is to prove redundancy of variables in subspaces and then ``merge'' the obtained results.
\slv~records the fact
that a set of variables $Z$ is redundant in \prob{X}{F} in subspace specified by partial assignment \pnt{q} as \Dss{X}{F}{r}{Z}.
Here \pnt{r} is a subset of the assignments of \pnt{q}  relevant to redundancy of $Z$. 
The record \Dss{X}{F}{r}{Z} is called a dependency sequent (or \tb{D-sequent} for short). To simplify notation, if $F$ and $X$ are obvious
from the context,  we record the D-sequent above   as just \Dds{r}{Z}.

A remarkable fact is that 
 a resolution-like operation called \ti{join}
can be used  to produce a new  D-sequent from two D-sequents derived earlier~\cite{tech_rep1,fmcad12}.  Suppose, for example,
that  D-sequents $(x_1=0,x_2=0) \rightarrow \s{x_9}$ and $(x_2=1,x_5=1) \rightarrow \s{x_9}$  specify redundancy of variable $x_9$
in different branches of  variable $x_2$. Then 
D-sequent $(x_1=0,x_5=1) \rightarrow \s{x_9}$ holds where the left part assignment of this D-sequent is obtained by
taking the union of the left part  assignments of the two D-sequents above but those to variable $x_2$.
The new D-sequent is said to be obtained by joining the two D-sequents above at variable $x_2$.
The calculus based on the join operation is complete. That is, eventually \slv~derives  D-sequent
$\emptyset\rightarrow X$ stating \ti{unconditional} redundancy of the variables of $X$ in \prob{X}{F}. If by the time
the D-sequent above is derived, $F$ contains an empty clause,
$F$ is unsatisfiable. Otherwise, $F$ is satisfiable. Importantly, if $F$ is satisfiable, 
derivation of D-sequent $\emptyset\rightarrow X$
does not require finding an assignment satisfying $F$.

DPLL-based SAT-solvers with clause learning can be viewed as a special case of \slv~where only a particular kind of D-sequents is used.
This limitation on D-sequents is  caused by the necessity to generate a satisfying assignment as a proof of satisfiability.
Importantly, this necessity deprives DPLL-based SAT-solvers of using transformations preserving equisatisfiability rather than
functional equivalence. In turn, this adversely affects the performance of SAT-solvers.
We illustrate this point by comparing the performance of DPLL-like SAT-solvers and a version of \slv ~on compositional formulas. 
This version of \slv~use the strategy of  lazy backtracking as opposed to that of  eager backtracking employed by DPLL-based procedures.
A compositional CNF formula 
has the form $F_1(X_1) \wedge \ldots \wedge F_k(X_k)$ where $X_i \cap X_j = \emptyset$,$i \neq j$. Subformulas $F_i, F_j$ are identical
modulo variable renaming/negation. We prove theoretically that performance of \slv~is \ti{linear} in $k$. On the other hand, one can argue that the
average performance of DPLL-based SAT-solvers with conflict learning should be \ti{quadratic} in $k$. In Section~\ref{sec:experiments},
we describe experiments confirming our theoretical results.

The contribution of this paper is fourfold. 
First, we use the machinery of D-sequents to explain some problems  of DPLL-based SAT-solvers.
Second, we describe a new QSAT-solver based on D-sequents called \slv.
Third, we give a theoretical analysis of the behavior of \slv~on compositional formulas.
Fourth,  we show the promise of~\slv by comparing its performance with that of well-known SAT-solvers on  compositional
and non-compositional formulas.

This paper is structured  as follows. In Section~~\ref{sec:qsat_sat} we discuss the complexity of QSAT and SAT.
 Section~\ref{sec:comparison} gives a brief introduction into \slv.
We recall D-sequent calculus in Section~\ref{sec:calculus}. A detailed description of \slv~is given in Section~\ref{sec:alg_description}.
Section~\ref{sec:comp_formulas} gives some theoretical results on performance of \slv. 
Section~\ref{sec:skipping} describes a modification of \slv~ that allows additional pruning of the search tree.
Experimental results are given
in Section~\ref{sec:experiments}. We describe some background of this research in Section~\ref{sec:background} and give
conclusions in Section~\ref{sec:conclusion}.

\section{Is QSAT Simpler Than SAT?}
\label{sec:qsat_sat}
%
%
\setlength{\intextsep}{4pt}
\setlength{\textfloatsep}{10pt}
\begin{wrapfigure}{L}{2in}
\small
\vspace{-10pt}
\begin{tabbing}
aaa\=bb\= cc\= ddddddd\= \kill
$\mi{gen\_sat\_assgn}(F)$\{\\
\tb{\scriptsize{1}}\> $\mi{ans}=\mi{solve\_qsat}(F)$; \\
\tb{\scriptsize{2}}\> if (\ti{ans}=\ti{unsat}) return(\ti{unsat}); \\  
\tb{\scriptsize{3}}\> $\pnt{s}:=\emptyset$; $X$:=\V{F}; \\
\tb{\scriptsize{4}}\> while ($X \neq \emptyset$) \{\\
\tb{\scriptsize{5}}\Tt  $v := \mi{pick\_var}(X)$; \\
\tb{\scriptsize{6}}\Tt  if ($\mi{solve\_qsat}(F_{v=0}) = \mi{sat}$) \\
\tb{\scriptsize{7}}\ttt   $\mi{val}=0$;\\
\tb{\scriptsize{8}} \Tt else  $\mi{val}=1$;\\
\tb{\scriptsize{9}}\Tt $F:=F_{v=val};$ \\
\tb{\scriptsize{10}}\Tt    $\pnt{s}=\pnt{s}\cup \s{(v=val)}$; \\
\tb{\scriptsize{11}}\Tt  $X:=X \setminus \s{v}$;\} \\ 
\tb{\scriptsize{12}}\> return(\pnt{s}); \\
\end{tabbing} 
\vspace{-20pt}
\caption{SAT-solving by QSAT}
\label{fig:sat_by_qsat}
\end{wrapfigure}

In this section, we make the following point. Both QSAT-solvers and SAT-solvers have exponential complexity on the set 
of \ti{all} CNF formulas, unless P = NP. However, this is not true for subsets of CNF formulas. It is possible that 
a set $K$ of formulas describing, say, properties of a parameterized set of designs can be solved in polynomial time by some QSAT-solver
while any SAT-solver has  exponential complexity on $K$.

To illustrate the point above, let us consider procedure \ti{gen\_sat\_assgn} shown in  Figure~\ref{fig:sat_by_qsat}.
It finds an assignment satisfying a CNF formula $F$ (if any) by solving a sequence of QSAT problems.
First, \ti{gen\_sat\_assgn} calls a QSAT-solver \ti{solve\_qsat} to check
 if $F$ is satisfiable (line 2). If it is, \ti{gen\_sat\_assgn} picks a variable $v$ of $F$ 
(line 5) and calls \ti{solve\_qsat} to find  assignment  $v=val$ under which formula $F$ is satisfiable (lines 6-8). Since $F$ is satisfiable,  
$F_{v=0}$ and/or $F_{v=1}$ has to be satisfiable. Then \ti{gen\_sat\_assgn} fixes variable $v$ at the chosen value \ti{val}  and 
adds ($v$=\ti{val}) to assignment \pnt{s}  (lines 9-10) that was originally empty. 
The \ti{gen\_sat\_assgn} procedure keeps assigning variables of $F$ in the same manner 
  in a loop (lines 5-11) until every variable of $F$ is assigned. At this point, \pnt{s} is a satisfying assignment of $F$.

The number of QSAT checks performed by \ti{solve\_qsat} in  \ti{gen\_sat\_assgn} is at most $n + 1$. 
So if there is a QSAT-solver solving all satisfiable CNF formulas in polynomial time, \ti{gen\_sat\_assgn} can use
this QSAT-solver in its inner loop to find a satisfying assignment for any satisfiable formula in polynomial time. However, this is not true
when considering \ti{a subset} $K$ of all possible CNF formulas.
Suppose there is a QSAT-solver solving the formulas of $K$ in polynomial time. Let $F$ be a formula of $K$. Let  \cof{F}{q} denote $F$ 
under partial assignment \pnt{q}.
The fact that $F \in K$ does not imply   $\cof{F}{q} \in K$. So the behavior of \ti{gen\_sat\_assgn} using this QSAT-solver
in the inner loop may actually be even \ti{exponential} if this QSAT-solver does not perform well on formulas \cof{F}{q}.

For example, one can  form a subset  $K$ of all possible CNF formulas such that a) a formula $F \in K$ describes a check that a number $N$ is composite
and b) an assignment satisfying $F$ (if any) specifies two numbers $A$,$B$ such that  $A \neq 1$,$B \neq 1$ and $A\times B = N$.
The satisfiability of formulas in $K$  can be checked by a QSAT-solver in polynomial time ~\cite{primes}. At the same time,
finding satisfying assignments of formulas from $K$ i.e. factorization of composite numbers is believed to be hard.
For instance,  \ti{gen\_sat\_assgn} cannot use the QSAT-solver above to find satisfying
assignments for formulas of $K$ in polynomial time.   The reason is that 
formula \cof{F}{q} does not specify  a check if a number is composite. That is $F \in K$ does not imply that $\cof{F}{q} \in K$.

Note that a SAT-solver is  also limited in the ways of  proving \ti{unsatisfiability}. For a SAT-solver, 
such a proof is just a failed attempt to build a satisfying assignment \ti{explicitly}. For example, instead of using
the polynomial algorithm of ~\cite{primes}, a SAT-solver would prove that a  number $N$ is prime by failing to find
two non-trivial factors of $N$.

\section{Brief comparison of DPLL-based SAT-solvers and \slv~in Terms of D-sequents}
\label{sec:comparison}
In this section, we use the notion of D-sequents to discuss some limitations of DPLL-based SAT-solvers. We also 
explain  how  \slv~ (described in Section~\ref{sec:alg_description} in detail) overcomes those limitations.
%
%
%
\begin{example}
\label{exmp:brief_comparison}
Let \cl~be a DPLL-based SAT-solver with clause learning. 
We assume that the reader is familiar with
the basics of such SAT-solvers~\cite{grasp,chaff}.
Let $F$ be a CNF formula of 8 clauses where $C_1=\overline{x}_1 \vee \overline{x}_3$,
$C_2 = \overline{x}_2 \vee x_3$, $C_3 = x_1 \vee x_2 \vee x_3$, $C_4 = x_2 \vee \overline{x}_3$,
$C_5 = \overline{x}_1 \vee x_4 \vee x_5$,  $C_6 = x_4 \vee \overline{x}_5$, $C_7 = \overline{x}_4 \vee x_5$, 
$C_8= \overline{x}_1 \vee \overline{x}_4 \vee \overline{x}_5$.
The set $X$ of variables of $F$ is equal to  \s{x_1,x_2,x_3,x_4,x_5}. 

Let \cl~first make assignment $x_1=0$. This satisfies clauses $C_1$,$C_5$,$C_8$ and removes literal $x_1$ from $C_3$.
Let \cl~then make assignment $x_2=0$.  Removing literal $x_2$ from  $C_3$ and $C_4$ turn them into unit clauses $x_3$ and $\overline{x}_3$
respectively. This means that \cl~ran into a conflict. At this point, \cl~generates  conflict clause $C_9 = x_1 \vee x_2$ that
is obtained by resolving clauses $C_3$ and $C_4$ on $x_3$ and adds $C_9$ to $F$. After that, \cl~erases assignment $x_2=0$ and the assignment made by \cl~to $x_3$
and runs BCP that assigns $x_2=1$ to satisfy $C_9$ that is currently unit.
In terms of D-sequents, one can view generation of conflict clause $C_9$ and adding it to $F$ as derivation of  D-sequent $S$ equal to
$(x_1=0,x_2=0)\rightarrow \s{x_3,x_4,x_5}$.  D-sequent  $S$ says
that making assignments falsifying clause $C_9$ renders all unassigned variables redundant. Note that $S$ is inactive
in the subspace $(x_1=0,x_2=1)$ that \cl~enters after assigning  1 to $x_2$. (We will say that D-sequent \Dds{r}{Z} is
\tb{active} in the subspace specified by partial assignment \pnt{q} if the assignments of \pnt{r} are a subset of those of \pnt{q}.)
So the variables $x_3,x_4,x_5$ 
proved redundant in subspace $(x_1=0,x_2=0)$ become non-redundant again.

One may think that reappearance of variables $x_3,x_4,x_5$ in subspace $(x_1=0,x_2=1)$  is ``inevitable'' but this is not so.
Variables $x_4$,$x_5$ have at least two reasons to be redundant in subspace $(x_1=0,x_2=0)$. First, 
 $C_9$ is falsified in this subspace. Second, the only clauses of $F$ containing variables $x_4$,$x_5$ are $C_5$,$C_6$,$C_7$,$C_8$.
But $C_5$ and $C_8$ are  satisfied by $x_1=0$ and  $C_6,C_7$
 can be satisfied  by an assignment to $x_4$,$x_5$.  So $C_5$,$C_6$,$C_7$,$C_8$  can be removed from $F$
in subspace $x_1=0$ without affecting the satisfiability of $F$. Hence D-sequents $S_1$ and $S_2$ equal to $(x_1=0)\rightarrow \s{x_4}$
and  $(x_1=0)\rightarrow \s{x_5}$ are true. (In  Example~\ref{exmp:alg_description}, we will show how  $S_1$ and $S_2$ are
derived by \slv.) Suppose that one replaces the D-sequent  $S$ above  with D-sequents $S',S_1,S_2$ 
where $S'$ is equal to $(x_1=0,x_2=0)\rightarrow \s{x_3}$.
Note that only D-sequent $S'$ is inactive in subspace $(x_1=0,x_2=1)$. So \ti{only variable} $x_3$ reappears after $x_2$ changes its
value from 0 to 1 $\square$
\end{example}

\vspace{-5pt}
The example above illustrates the main difference between \cl~and \slv~in terms of D-sequents.
At every moment,  \cl~has \ti{at most one}  active D-sequent.
This D-sequent is of the form \Dds{r}{Z} 
where $r$ is an assignment falsifying a clause of $F$ and $Z$ is the set of \ti{all} variables that are currently
unassigned. \slv~may have a \ti{set} of active D-sequents $\Dds{r_1}{Z_1},\ldots,\Dds{r_k}{Z_k}$
 where $Z_1 \cup \ldots \cup Z_k = Z$, $Z_i \cap Z_j=\emptyset$,$i \neq j$. When \cl~changes the value of variable $v$ of \Va{r},
all the variables of $Z$ reappear as non-redundant. When \slv~changes the value of $v$,  variables
of $Z_i$ reappear \ti{only} if $v \in \Va{r_i}$. So only a subset of variables of $Z$ reappear.

 To derive D-sequents \Dds{r_i}{Z_i} above, \slv~goes on branching
\ti{in the presence of a conflict}. Informally, the goal of such branching is to find alternative ways of proving redundancy of  variables from $Z$.
So \slv~uses extra branching to minimize the number of variables reappearing in the right branch (after the left
branch has been explored). This should eventually lead to the opposite result i.e. to reducing  the amount of branching. 
Looking for alternative ways to prove redundancy
can be justified as follows. A practical formula $F$ typically can be represented as $F_1(X_1,Y_1) \wedge \ldots \wedge F_k(X_k,Y_k)$. 
Here $X_i$ are internal variables of $F_i$ and $Y_i$ are ``communication'' variables
that $F_i$ may share with some other subformulas $F_j$, $j \neq i$. One can view $F_i$ as describing a ``design block'' with external
variables $Y_i$. Usually, $|Y_i|$ is much smaller than $|X_i|$.  Let a clause of  $F_i$ be falsified by the current assignment
due to a conflict. Suppose that 
at the time of the conflict all variables of $Y_j$ of subformula $F_j$ were assigned and their values were specified by assignment \pnt{y_i}.
Suppose \pnt{y_i} is consistent for $F_i$ i.e. \pnt{y_i} can be extended by assignments to $X_i$ to satisfy $F_i$. This means
that the variables of $X_i$ are redundant in subspace \pnt{y_i} in \prob{V}{F} where $V = \V{F}$. Then by branching
on variables of $X_i$ one can derive D-sequent \Dds{y_i}{X_i}. If \pnt{y_i} is inconsistent for $F_i$, then
by branching on variables of $X_i$ one can derive a clause $C$ falsified by \pnt{y_i}. Adding $C$ to $F$ makes the variables of $X_i$
redundant in \prob{V}{F} in subspace \pnt{y_i}.  So the existence of many ways to prove variable redundancy is essentially implied by the fact
that formula $F$ has  structure.

The possibility to  control the size of right branches  gives an algorithm a lot of power.
Suppose, for example, that an algorithm guarantees that  the number of variables reappearing in the right branch is bounded by
a constant $d$. We assume that this applies to the right branch going out of any node of the search tree, including the root node.
 Then the size of the search tree built by such an algorithm is  $O(|X|\cdot 2^d)$. Here $|X|$ is the maximum depth 
of a search tree built by branching on variables of $X$  and $2^d$ is the number of nodes in a full binary sub-tree over $d$ variables. So the factor
$2^d$ limits the size of the right branch. The complexity of an algorithm building such a search tree is 
 \ti{linear} in $F$. In Section~\ref{sec:comp_formulas}, we show that  bounding the size
of right branches by a constant is exactly the reason why the complexity of \slv~on compositional formulas is linear in the number
of subformulas.

The limitation of
D-sequents available to \cl~ is consistent with the necessity to produce a satisfying assignment. Although such limitation
cripples the ability of an algorithm to compute  the parts of the formula that are redundant in the current subspace, it does not matter
much for \cl. 
The latter simply cannot use  this redundancy because it is formulated with respect to formula \prob{X}{F} rather than $F$.
Hence, discarding the clauses
containing redundant variables preserves equisatisfiability rather than functional equivalence.
So, an algorithm using such transformations cannot guarantee that  a satisfying assignment it found is correct.

\section{D-sequent Calculus}
\label{sec:calculus}
In this section, we recall the   D-sequent calculus introduced ~\cite{tech_rep1,fmcad12}.
In Subsections~\ref{subsec:basic_def} and ~\ref{subsec:simple_cases} we give basic definitions
and describe simple cases of variable redundancy.
 The notion of D-sequents is introduced in Subsection~\ref{subsec:d_seqs}. Finally,  the operation 
of joining D-sequents is presented in Subsection~\ref{subsec:join_oper}.
\label{sec:rvars_bps}
%
%
\subsection{Basic definitions}
\label{subsec:basic_def}
%
%
\begin{definition}
\label{def:cnf}
A \tb{literal} of a Boolean variable $v$ is $v$ itself and its negation.
A \tb{clause} is a disjunction of literals.  A formula $F$ represented
as a conjunction of clauses is said to be the Conjunctive Normal Form (\tb{CNF}) of $F$.
A CNF formula $F$ is also viewed as \tb{a set of clauses}.
 Let \pnt{q} be an
assignment, $F$ be a CNF formula, and $C$ be a clause. {\boldmath \Va{q}}
denotes the variables assigned in \pnt{q}; {\boldmath \V{F}} denotes the set
of variables of $F$; {\boldmath \V{C}} denotes the set of variables of $C$.
\end{definition}
%
%
\begin{definition}
\label{def:cofactor}
Let \pnt{q} be an assignment. Clause $C$ is \tb{satisfied}  by \pnt{q} if a literal of $C$ is
set to 1 by \pnt{q}. Otherwise, $C$ is \tb{falsified} by \pnt{q}. Assignment \pnt{q}
\tb{satisfies} $F$ if \pnt{q} satisfies every clause of $F$. 
\end{definition}
%
%
\begin{definition}
Let $F$ be a CNF formula and \pnt{q} be a partial assignment to variables of $F$.
Denote by \cof{F}{q} that is obtained from $F$ by 
a) removing all clauses of $F$ satisfied by \pnt{q};
b) removing the literals set to 0 by \pnt{q} from the clauses that are not satisfied by \pnt{q}.
Notice, that if \pnt{q}=$\emptyset$, then \cof{F}{q} = $F$.  
\end{definition}

%
%
\begin{definition}
\label{def:X_clause}
Let $F$ be a CNF formula and $Z$ be a subset of  \V{F}.
 Denote by {\boldmath $F^{Z}$} the set of all clauses of $F$ containing at least one  variable of $Z$.
\end{definition}
%
%
\begin{definition}
\label{def:red_vars}
The variables of $Z$ are 
\textbf{redundant}  in formula \prob{X}{F} if
\prob{X}{F} $\equiv \exists X [F\setminus F^Z]$. We note that 
since $F \setminus F^Z$ does not contain any $Z$ variables, we
could have written  $\exists (X \setminus Z) [F \setminus
F^Z]$. To simplify notation, we avoid explicitly using this
optimization in the rest of the paper. 
\end{definition}

%
%
\begin{definition}
Let \pnt{q_1} and \pnt{q_2} be assignments. The expression
$\pnt{q_1} \leq \pnt{q_2}$ denotes the fact that $\Va{q_1} \subseteq \Va{q_2}$ and each variable  of \Va{q_1}
has the same value in \pnt{q_1} and \pnt{q_2}.
\end{definition}

%
%

\subsection{Simple cases of variable redundancy}
\label{subsec:simple_cases}
There at least two cases where proving that a variable of $F$ is redundant in \prob{X}{F} is easy. The first case concerns
monotone variables of $F$. A variable $v$ of $F$ is called \tb{monotone}  
if all clauses of $F$ containing $v$ have only positive (or only negative) literal of $v$.
A monotone variable $v$ is redundant in \prob{X}{F} because removing  the  clauses with $v$ from $F$ does not
change the satisfiability of $F$.
The second case concerns the presence of an empty clause. If $F$ contains such a clause, every  variable of $F$
is redundant.

%
%
\subsection{D-sequents}
\label{subsec:d_seqs}
%
%
%
\begin{definition}
\label{def:d_sequent}
Let $F(X)$  be a CNF formula. Let \pnt{q} be an assignment to $X$
and $Z$ be a subset of $X \setminus \Va{q}$.
A dependency sequent (\textbf{D-sequent})  has the form
$(\prob{X}{F},\pnt{q})\!\rightarrow\!Z$. It states that 
the variables of $Z$ are redundant in  \prob{X}{\cof{F}{q}}.
If formula $F$ for which a D-sequent holds
 is obvious from the context we will write this D-sequent in a short notation:
\Dds{q}{Z}\!\!.
\end{definition}

%
%
\begin{example}
Let $F$ be a CNF formula of  four clauses: $C_1 = x_1 \vee x_2$,
$C_2 = \overline{x}_1 \vee \overline{x}_2$, $C_3 = \overline{x}_1 \vee x_3$,
$C_4 = x_2 \vee \overline{x}_3$. Notice that since clause $C_1$ is satisfied 
in subspace $(x_2=1)$,  variable $x_1$ is monotone in formula $F_{x_2=1}$.
So D-sequent $(x_2=1) \rightarrow \s{x_1}$ holds. On the other hand, the assignment $\pnt{r}=(x_1=1,x_3=0)$ falsifies
clause $C_3$. So variable $x_2$ is redundant in \cof{F}{r} and D-sequent \Dds{r}{\s{x_2}} holds.
\end{example}

%
%
\subsection{Join Operation for D-sequents}
\label{subsec:join_oper}

%
%
\begin{proposition}[\cite{tech_rep1}]
\label{prop:join_rule}
Let $F(X)$  be a CNF formula. Let D-sequents \Dds{r'}{Z} and \Dds{r''}{Z}  hold, where $Z \subseteq X$.
 Let \pnt{r'}, \pnt{r''} have different values for exactly one variable
$v \in \Va{r'} \cap \Va{r''}$. Let \pnt{r} consist of all
assignments of \pnt{r'},\pnt{r''} but those to $v$.
Then, D-sequent  \Dds{r}{Z}  holds too.
\end{proposition}
%
%

\vspace{-5pt}
We will say that the D-sequent \Dds{r}{Z} of Proposition~\ref{prop:join_rule} is obtained
by \tb{joining D-sequents} \Dds{r'}{Z} and \Dds{r''}{Z} at variable $v$. 
The join operation is \ti{complete}~\cite{tech_rep1,fmcad12}. That is eventually, D-sequent $\emptyset \rightarrow X$
is derived proving that the variables of the current formula $F$ are redundant. If $F$ contains an empty clause,
then $F$ is unsatisfiable. Otherwise, it is unsatisfiable.

An obvious difference between the D-sequent calculus and resolution is that the former can handle both satisfiable
and unsatisfiable formulas. This limitation of resolution is due to the fact that it operates on subspaces
where formula $F$ is unsatisfiable. One can interpret resolving clauses $C',C''$ to produce clause $C$ as 
using the Boolean cubes $K'$,$K''$ where $C'$ and $C''$ are unsatisfiable to produce a new Boolean cube $K$ where
the resolvent $C$ is unsatisfiable. On the contrary, the join operation can be performed over parts of the search space where $F$
may be satisfiable. When D-sequents  \Dds{r'}{Z} and \Dds{r''}{Z} are joined, it does not matter whether formulas \cof{F}{r'}
and \cof{F}{r''} are satisfiable. The only thing that matters is that variables $Z$ are redundant in \cof{F}{r'} and \cof{F}{r''}.

%
%
\subsection{Virtual redundancy}
\label{subsec:virt_redund}
Let $F(X)$ be a CNF formula and \pnt{r}  be an assignment to $X$.
Let $Z \subseteq X$ and $\Va{r} \cap Z = \emptyset$.
 The fact that variables of $Z$ are redundant in $F$, in general,
does not mean that they are redundant in \cof{F}{r}.  Suppose, for example, that 
$F$ is satisfiable, \cof{F}{r} is unsatisfiable, $F$ does not have a clause falsified by \pnt{r} and $Z = \V{F} \setminus \Va{r}$.
Then  formula $\cof{F}{r} \setminus (\cof{F}{r})^Z$ has no clauses and so is satisfiable. Hence $\prob{X}{\cof{F}{r}} \neq $ 
$\exists{X}[\cof{F}{r} \setminus (\cof{F}{r})^Z]$ and so the variables of $Z$ are not redundant in \cof{F}{r}. On the other hand,
since $F$ is satisfiable,  the variables of $Z$ are redundant in \prob{X}{F}. 

We will say that the variables of $Z$ are \tb{virtually redundant} in \cof{F}{r} where $Z \cap \Va{r} = \emptyset$ if either a) 
$\prob{X}{\cof{F}{r}} = \exists{X}[\cof{F}{r} \setminus (\cof{F}{r})^Z]$ or b) 
$\prob{X}{\cof{F}{r}} \neq \exists{X}[\cof{F}{r} \setminus (\cof{F}{r})^Z]$ and $F$ is satisfiable.
In other words, if variables $Z$ are virtually redundant in \prob{X}{\cof{F}{r}},  removing the clauses with
a variable of $Z$ from \cof{F}{r} may be wrong but only \ti{locally}. From the global point of view 
this mistake does not matter because it occurs only when $F$ is satisfiable.

We need a new notion of redundancy because the join operation introduced above
preserves  virtual redundancy~\cite{tech_rep1} rather than redundancy in terms of Definition~\ref{def:red_vars}. Suppose, for example, 
that the  variables of $Z$ are redundant in \cof{F}{r_1} and \cof{F}{r_2} in terms of Definition~\ref{def:red_vars}
and so D-sequents  \Dds{r_1}{Z} and \Dds{r_2}{Z} hold.
  Let \Dds{r}{Z} be the D-sequent obtained by  joining the  D-sequents above. Then one can guarantee only
that the variables of $Z$ are virtually redundant in \cof{F}{r}. For that reason we need to replace the notion of redundancy
by Definition~\ref{def:red_vars} with that of  virtually redundancy.
 In the future explanation,  we will omit the word ``virtually''.
 That is  when  we say that variables of $Z$ are redundant in \cof{F}{r}
we actually mean that they are \ti{virtually} redundant in \cof{F}{r}. 

\section{Description of \slv}
\label{sec:alg_description}
In this section, we describe \slv, a QSAT-solver based on the machinery of D-sequents. 
%
%
\setlength{\intextsep}{2pt}
\setlength{\textfloatsep}{2pt}
\begin{wrapfigure}{L}{2.3in}
\small
// $F$  is a CNF formula \\
// \pnt{q} is an assignment to \V{F} \\
// \DS~is a set of active D-sequents \\
\vspace{-10pt}
\begin{tabbing}
aaaa\=bb\=cc\=dd\= \kill
\Slv($F$,\pnt{q},\DS)\{\\
\tb{\scriptsize{1}}\>  if (\ti{empty\_clause}($F$)) \\
                   \Tt    exit(\ti{unsat});\\
\tb{\scriptsize{2}}\>   if ($\mi{new\_falsif\_clause}(C,F,\pnt{q})$)\\
\tb{\scriptsize{3}}\Tt     if (\ti{left\_branch}(\pnt{q})) \\ 
\tb{\scriptsize{4}}      \ttt  \DS:=$\mi{update\_Dseqs}(\DS,F,C)$;\\
\tb{\scriptsize{5}}\Tt  else \{ \\
\tb{\scriptsize{6}}       \ttt  \DS:=$\mi{finish\_Dseqs}(\DS,F,C)$;\\
\tb{\scriptsize{7}}           \ttt return($F,\DS$); \}\\
\tb{\scriptsize{8}}\> $\DS := \mi{monot\_vars\_Dseqs}(\DS,F,\pnt{q})$;\\
\tb{\scriptsize{9}}\> if (\ti{all\_vars\_assgn\_or\_redund}(\DS,\pnt{q}); \\
\tb{\scriptsize{10}}\Tt if ($\mi{no\_falsif\_clauses}(F,\pnt{q})$) \\
                       \ttt  exit(\ti{sat}); \\
\tb{\scriptsize{11}}\Tt else  return($F,\DS$); \\
              \> - - - - - - - - - - - - - -\\
\tb{\scriptsize{12}}\> $v := \mi{pick\_variable}(F,\pnt{q},\DS)$; \\
\tb{\scriptsize{13}}\> \pnt{q_0}=\pnt{q} $\cup$   \s{(v=0)}; \\
\tb{\scriptsize{14}}\> $(F,\DS_0) \leftarrow$\Slv($F$,\DS,\pnt{q_0});\\
\tb{\scriptsize{15}}\> $(\Sup{\DS}{sym},\Sup{\DS}{asym}) = \mi{split}(\DS_0,v)$;\\
\tb{\scriptsize{16}}\> if ($\Sup{\DS}{asym} = \emptyset$)  return($F,\DS_0$);\\
\tb{\scriptsize{17}}\> $\mi{recover\_vars\_clauses}(F,\Sup{\DS}{asym})$; \\
\tb{\scriptsize{18}}\> \pnt{q_1}=\pnt{q} $\cup$   \s{(v=1)}; \\

\tb{\scriptsize{19}}\> $(F,\DS_1) \leftarrow$\Slv($F$,$\Sup{\DS}{sym}$,\pnt{q_1});\\
             \> - - - - - - - - - - - - - - \\
\tb{\scriptsize{20}}\> $(F,\DS)\!\leftarrow\!\mi{merge}(F,v,\pnt{q},\DS_0,\DS_1)$;\\
\tb{\scriptsize{21}}\> return($F,\DS$);\} \\
\end{tabbing} 
\vspace{-25pt}
\caption{\slv~procedure}
\label{fig:high_level_descr}
\end{wrapfigure}

%
%
\subsection{High-level view}
Pseudocode of ~\slv~is given in Figure~\ref{fig:high_level_descr}. \slv~accepts  a CNF formula $F$, 
a partial assignment \pnt{q} to $X$ where $X = \V{F}$, and a set of active D-sequents 
\DS~stating redundancy of \ti{some} variables from $X \setminus \Va{q}$ in subspace \pnt{q}.
\slv~returns CNF formula $F$ that consists of the clauses of the initial formula plus some resolvent clauses and a set \DS~of D-sequents
stating redundancy of \ti{every} variable of $X \setminus \Va{q}$ in subspace \pnt{q}. 
To check satisfiability of a CNF formula, one needs to call \slv~with $\pnt{q}=\emptyset$,
$\DS = \emptyset$.

\slv~is a branching procedure. If \slv~cannot prove redundancy of some variables in the current subspace, it picks
one of such variables $v$  and branches on it. So \slv~ builds a binary search tree where a node corresponds to a branching variable.
 We will refer to the first (respectively second) assignment  to $v$
as the \tb{left}  (respectively \tb{right}) \tb{branch} of $v$. Although Boolean Constraint Propagation (BCP) is not explicitly mentioned
in Figure~\ref{fig:high_level_descr},
it is included into the \ti{pick\_variable} procedure as follows. Let \pnt{q} be the current partial assignment. Then
a) preference is given  to branching on variables of unit clauses of \cof{F}{q} (if any); b) if $v$ is a variable of a unit clause of $C$ of \cof{F}{q}
and $v$ is picked for branching, then the value satisfying $C$ is assigned first.

As soon as a variable $v$ is proved redundant in the current subspace \pnt{q}, a D-sequent \Dds{r}{\s{v}} is recorded where
\pnt{r} is a subset of assignments of \pnt{q}. All the clauses of $F$ containing variable $v$ are marked as redundant and \ti{ignored}
until $v$ becomes non-redundant again. This happens when a variable of \Va{r} changes its value making the D-sequent \Dds{r}{\s{v}}
inactive in the current subspace.

As we mentioned in Section~\ref{sec:comparison}, if a clause $C$ 
containing a variable $v$ is falsified after an assignment is made to $v$, \slv~keeps making assignments to unassigned non-redundant variables.
However, this happens only in the \ti{left} branch of $v$. If $C$ is falsified in the right branch of $v$, \slv~backtracks. 
A unit clause $C'$  gets  falsified in the left branch only when
\slv~tries  to satisfy another unit clause $C''$ such that $C'$ and $C''$ have the opposite literals of a variable $v$.
We will refer to the node of the search tree corresponding
to $v$ as a \ti{conflict} one. The number of conflict nodes \slv~may have is not limited.

\slv~consists of three parts. In Figure~\ref{fig:high_level_descr}, they are separated by dashed lines. In the first part, described
in Subsections~\ref{subsec:termination} and~\ref{subsec:atomic_dseqs} in more detail, \slv~checks for termination conditions
and builds D-sequents for variables whose redundancy is obvious. In the second part (Subsection~\ref{subsec:branching}),
\slv~picks an unassigned non-redundant variable $v$ and splits the current subspace into  subspaces $v=0$ and $v=1$.
Finally, \slv~merges the results of branches $v=0$ and $v=1$ (Subsection~\ref{subsec:merging}).

%
%
\subsection{Eager and lazy backtracking (DPLL as a special case of \slv)}
\label{subsec:eager_lazy}
Let \pnt{q} be the current partial assignment to variables of $X$ and variable $v$ be the  variable assigned in \pnt{q} most recently.
Let $v$ be assigned a first value (left branch).  Let $C$ be a clause of $F$ falsified after $v$ is assigned in \pnt{q}.
In this case, procedure \ti{update\_Dseqs} of \slv~(line 4 of Figure~\ref{fig:high_level_descr}), 
derives a D-sequent \Dds{r}{Z'}. Here \pnt{r} is the smallest subset of assignments of \pnt{q} falsifying $C$ and 
$Z'$ is a subset of the current set $Z$ of the unassigned, non-redundant variables.

The version  where $Z' = \emptyset$ i.e. where no D-sequent \Dds{r}{Z'} is derived by \ti{update\_Dseqs}  will be called \slv~with \ti{lazy backtracking}.
In our theoretical and experimental evaluation of \slv~given in Sections~\ref{sec:comp_formulas}
 and~\ref{sec:experiments} we used the version with lazy backtracking.
The version of \slv~where  $Z'$ is always equal to $Z$ will be referred to as \slv~with \ti{eager backtracking}. 
DPLL is a special case of \slv~where the latter employs eager backtracking. In this case, all unassigned variables are declared
redundant and \slv~immediately backtracks without trying to prove redundancy of variables of $Z$ in some other ways.

%
%
\setlength{\intextsep}{0pt}
\setlength{\textfloatsep}{0pt}
\begin{wrapfigure}{L}{2in}
\small
\begin{tabbing}
aaaa\=bb\=cc\ dd\= \kill
// \pnt{q_0}=\pnt{q}$\cup$\s{(v=0)}; \pnt{q_1}=\pnt{q}$\cup$\s{(v=1)}; \\
// $C_0 = \mi{nil}$, $C_1=\mi{nil}$ if no clause of $F$ \\
// is falsified by \pnt{q_0},\pnt{q_1} respectively \\
  \\
$\mi{merge}(F,v,\pnt{q},\DS_0,\DS_1)$\{\\
\tb{\scriptsize{1}}\> for ($w \in (\V{F} \setminus (\Va{q} \cup {v})$) \{ \\
\tb{\scriptsize{2}}\Tt if ($\mi{symmetric\_in\_v}(\DS_1,w)$) \\
                    \ttt continue;\\
\tb{\scriptsize{3}}\Tt $S_0 = extract\_Dseq(\DS_0,w)$; \\
\tb{\scriptsize{4}}\Tt $S_1 = extract\_Dseq(\DS_1,w)$; \\
\tb{\scriptsize{5}}\Tt $S = join(S_0,S_1,v)$; \\
\tb{\scriptsize{6}}\Tt $\DS_1 = (\DS_1 \cup \s{S}) \setminus \s{S_1}$ ;\}\\
\> - - - - - - - - - - - - - - -  \\
\tb{\scriptsize{7}}\> $C_0 = pick\_falsif\_clause(F,\pnt{q_0})$; \\
\tb{\scriptsize{8}}\> $C_1 = pick\_falsif\_clause(F,\pnt{q_1})$; \\
\tb{\scriptsize{9}}\> if (($C_0 \neq \mi{nil}$) and ($C_1 \neq \mi{nil}$)) \{\\
\tb{\scriptsize{10}}   \Tt  $C = \mi{resolve}(C_0,C_1,v)$;  \\
\tb{\scriptsize{11}}   \Tt  $F = F \cup \s{C}$; \\
\tb{\scriptsize{12}}    \Tt  $\DS_1 = \DS_1\!\cup\!\s{\mi{falsif\_clause\_Dseq}(C,v)}$; \\
\tb{\scriptsize{13}}\> else \\
\tb{\scriptsize{14}}   \Tt $\DS_1 = \DS_1\!\cup\!\s{\mi{monot\_var\_Dseq}(F,v,\pnt{q})}$; \\
\tb{\scriptsize{15}}\> return($F,\DS_1$); \}\\
\end{tabbing} 
\vspace{-25pt}
\caption{\ti{merge} procedure}
\label{fig:merge}
\end{wrapfigure}

%
%
\subsection{Termination conditions}
\label{subsec:termination}
\slv~reports unsatisfiability if the current formula $F$ contains an empty clause (line 1 of Figure~\ref{fig:high_level_descr}).
\slv~reports satisfiability if no clause of $F$ is falsified by the current assignment \pnt{q} and
every variable of $F$ is either assigned in \pnt{q} or proved redundant in subspace \pnt{q} (line 10). Note that \slv~uses slight optimization
here by terminating before  
the  D-sequent $\emptyset \rightarrow X$ is derived stating unconditional redundancy of variables of $X$ in \prob{X}{F}.

If no termination condition is  met but  every variable of $F$ is assigned or proved redundant,
\slv~ends the current call and returns  $F$ and \DS~(lines 7,11).  
In contrast to operator \ti{return}, the operator \ti{exit} used in lines 1,10 eliminates
the entire stack of nested calls of \slv.

%
%
\subsection{Derivation of atomic D-sequents}
\label{subsec:atomic_dseqs}
Henceforth, for simplicity, we will assume that \slv~ derives \tb{D-sequents of the form} {\boldmath \Dds{r}{\s{v}}} i.e.
for single variables.
A D-sequent \Dds{r}{Z} is then  represented as $|Z|$ different D-sequents \Dds{r}{\s{v}}, $v \in Z$. 

In the two cases below, variable redundancy is obvious. Then \slv~derives D-sequents we will call \tb{atomic}.
The first case, is when clause of $F$ is falsified by the current assignment \pnt{q}. This kind of D-sequents is 
derived by procedures \ti{update\_Dseqs} (line 4) and \ti{finish\_Dseqs}(line 6). Let $v$ be the variable
assigned in \pnt{q} most recently. Let $C$ be a clause of $F$ falsified after the current assignment to $v$ is made.
If $v$ is assigned a first value (left branch), then, as we mentioned in Subsection~\ref{subsec:eager_lazy},
for \ti{some} unassigned variables $w_1,\ldots,w_m$ that are not proved redundant yet, one can build D-sequents 
\Dds{r}{\s{w_1}},...,\Dds{r}{w_m}. 
Here \pnt{r} is the shortest assignment falsifying $C$. So \ti{update\_Dseqs} may leave some unassigned variables 
non-redundant.
 On the contrary, \ti{finish\_Dseqs} is called in the right branch of $v$. In this case, for \ti{every} unassigned variable $w_i$
that is not proved redundant yet, D-sequent \Dds{r}{\s{w_i}} is generated.
 So on exit from \ti{finish\_Dseqs}, every variable of $F$ is
either assigned or proved redundant.

D-sequents of monotonic variables are the second case of atomic D-sequents.  They are  generated by procedure \ti{monot\_vars\_Dseqs}
 (line 8) and by procedure \ti{monot\_var\_Dseq} called when \slv~merges results of branches 
 (line 14 of Figure~\ref{fig:merge}). Let \pnt{q} be the current partial assignment and $v$ be a monotone unassigned variable of $F$.
Assume for the sake of clarity, that only clauses with positive polarity of $v$ are present in \cof{F}{q}. This means
that every clause of $F$ with literal $\overline{v}$ is either satisfied by \pnt{q} or contains  a variable $w$ proved
redundant in \cof{F}{q}. Then \slv~generates D-sequent \Dds{r}{\s{v}} where $\pnt{r}$ is formed from assignments of \pnt{q} as follows. 
For every clause $C$ of $F$ with 
literal $\overline{v}$  assignment \pnt{r}  a) contains an assignment satisfying $C$ or b)  contains all the assignments
of \pnt{s} such that D-sequent \Dds{s}{\s{w}} is active and $w$ is a variable of $C$. Informally, \pnt{r} contains a set of assignments
under which variable $v$ becomes monotone.
%
%
\subsection{Branching in \slv}
\label{subsec:branching}

When \slv~cannot prove redundancy of some unassigned variables in the current subspace \pnt{q}, it picks
a non-redundant variable $v$  for branching (line 12 of Figure~\ref{fig:high_level_descr}).  
First, \slv~ calls itself with assignment $\pnt{q_0}=\pnt{q} \cup \s{(v=0)}$. 
(Figure~\ref{fig:high_level_descr}  shows the case when 
assignment $v=0$ is explored in the left branch but obviously the assignment $v=1$ can be explored before $v=0$.)
Then \slv~partitions the returned  set of D-sequents $\DS_0$ into \Sup{\DS}{sym} and \Sup{\DS}{asym}.

The set \Sup{\DS}{sym} consists of the D-sequents \Dds{r}{\s{w}} of $\DS_0$ such that $v \not\in \Va{r}$. The D-sequents
of \Sup{\DS}{sym} remain active in the branch $v=1$. The set \Sup{\DS}{asym} consists of the
D-sequents \Dds{r}{\s{w}} such that \pnt{r} contains assignment $(v=0)$. The D-sequents of \Sup{\DS}{asym} are inactive
in the subspace $v=1$ and the variables whose redundancy is stated by those D-sequents reappear in the right branch.
If  $\Sup{\DS}{asym}=\emptyset$, there is no reason to  explore the right branch. So, \slv~just returns the set of D-sequents $\DS_0$
(line 16). Otherwise, \slv~recovers the variables and clauses that were marked redundant after D-sequents from  \Sup{\DS}{asym} were derived (line 17)
and calls itself with partial assignment $\pnt{q_1}=\pnt{q} \cup \s{(v=1)}$.
%
%
\subsection{Merging results of branches}
\label{subsec:merging}
After both branches of variable $v$ has been explored, \slv~merges the results by calling the \ti{merge} procedure (line 20).
The pseudocode of \ti{merge} is shown in Figure~\ref{fig:merge}. \slv~backtracks only when every unassigned variable is proved
redundant in the current subspace. The objective of \ti{merge} is to maintain this invariant by a) replacing
the currently D-sequents that depend on the branching variable $v$ with those that are symmetric in $v$;
b) building  a D-sequent for the branching variable $v$ itself.

The \ti{merge} procedure consists of two parts separated in Figure~\ref{fig:merge} by the dotted line. In the first part, \ti{merge}
builds  D-sequents for the variables of  $X \setminus (\Va{q} \cup \s{v})$. In the second part, it builds
a D-sequent for the branching variable. In the first part, \ti{merge} iterates over variables $X \setminus (\Va{q} \cup \s{v})$.
Let $w$ be a variable of $X \setminus (\Va{q} \cup \s{v})$. If the current D-sequent for $w$ (i.e. the D-sequent for $w$ from the set $\DS_1$
returned in the right branch) is symmetric in $v$, then there is no need to build a new D-sequent (line 2).  Otherwise, a new D-sequent $S$
for $w$ that does not depend on $v$ is generated as follows. Let $S_0$ and $S_1$ 
be the D-sequents  for variable $w$ contained in $\DS_0$ and $\DS_1$ respectively (lines 3,4). That is $S_0$ and $S_1$ were
generated for variable $w$ in branches $v=0$ and $v=1$.  Then  D-sequent $S$  is produced
by joining $S_0$ and $S_1$ at variable $v$ (line 5).

\setlength{\intextsep}{4pt}
\begin{wrapfigure}{l}{2.5in}
 \begin{center}
    \includegraphics{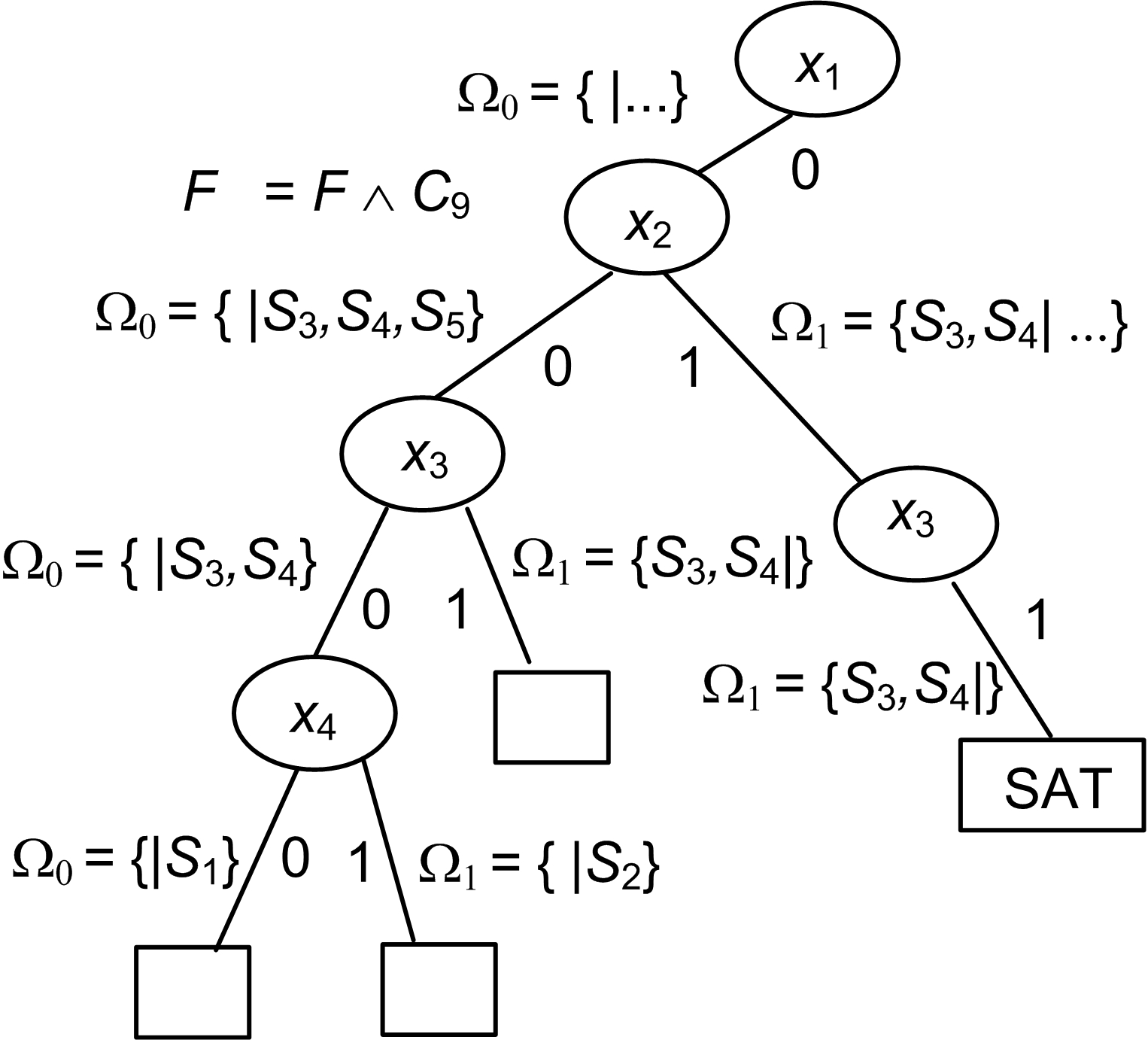}
  \end{center}
\caption{Search tree built by \slv}
\label{fig:tree}
\end{wrapfigure}

Generation of  a D-sequent for the variable $v$ itself depends on whether node $v$ (i.e the node of the search tree corresponding to $v$) 
is a conflict one. If so, $F$ contains clauses $C_0$ and $C_1$ that have variable $v$ 
and are falsified by \pnt{q_0} and \pnt{q_1} respectively.  In this case, to make variable $v$ redundant \ti{merge} generates
the resolvent $C$ of $C_0$ and $C_1$ on variable $v$ and adds $C$ to $F$ (lines 10,11).  Then D-sequent \Dds{r}{\s{v}} is generated where
\pnt{r} is the shortest assignment falsifying clause $C$ (line 12). 

%
%
\begin{wrapfigure}{L}{2in}
\small
\begin{tabbing}
aa\=bb\=cc\ dd\= \kill
$S_1:$ \Tt $(x_1=0,x_4=0) \rightarrow \s{x_5}$ \\
$S_2:$ \Tt $(x_1=0,x_4=1) \rightarrow \s{x_5}$ \\
$S_3:$ \Tt $(x_1=0) \rightarrow \s{x_5}$ \\
$S_4:$ \Tt $(x_1=0) \rightarrow \s{x_4}$ \\
$S_5:$ \Tt $(x_1=0,x_2=0) \rightarrow \s{x_3}$ \\
\end{tabbing} 
\vspace{-25pt}
\caption{D-sequents of Figure~\ref{fig:tree}}
\label{fig:derived_dseqs}
\end{wrapfigure}

If node $v$ is not a conflict one, this means that clause $C_0$ and/or clause $C_1$ does not exist.
Suppose, for example, that no clause $C_0$   containing variable  $v$ is falsified by \pnt{q_0}.
This means that every clause $F$ with the positive literal of $v$ is either satisfied by \pnt{q}
or contains a variable redundant in subspace \pnt{q}. In other words, $v$ is monotone in \cof{F}{q} after removing
the clauses with redundant variables.
Then an atomic D-sequent is generated by \ti{merge} (line 14) as described in Subsection~\ref{subsec:atomic_dseqs}.

\begin{example}
\label{exmp:alg_description}
Here we show how \slv~with \ti{lazy backtracking} operates when solving the CNF formula $F$ introduced in Example~\ref{exmp:brief_comparison}.
Formula $F$  consists of 8 clauses:  $C_1=\overline{x}_1 \vee \overline{x}_3$,
$C_2 = \overline{x}_2 \vee x_3$, $C_3 = x_1 \vee x_2 \vee x_3$, $C_4 = x_2 \vee \overline{x}_3$,
$C_5 = \overline{x}_1 \vee x_4 \vee x_5$,  $C_6 = x_4 \vee \overline{x}_5$, $C_7 = \overline{x}_4 \vee x_5$, 
$C_8= \overline{x}_1 \vee \overline{x}_4 \vee \overline{x}_5$. Figure~\ref{fig:tree} shows the search tree built by
\slv. The 
ovals specify the branching nodes labeled by the corresponding  branching variables. The label 0 or 1 on the edge connecting two nodes
specifies the value made to the variable of the higher node. The rectangles specify the leaves of the search tree.
The rectangle SAT specifies the leaf where \slv~reported that $F$ is satisfiable.

Every edge of the search tree labeled with value 0 (respectively 1)  also shows the set of D-sequents $\DS_0$ (respectively $\DS_1$)
derived when the assignment corresponding to this edge was made. The D-sequents produced by \slv~ are denoted in Figure~\ref{fig:tree} as $S_1,\ldots,S_5$.
The values of $S_1,\ldots,S_5$ are given  in Figure~\ref{fig:derived_dseqs}. When representing $\DS_0$ and $\DS_1$, we use the  symbol '$|$' to
separate D-sequents derived before and after a call of \slv. Consider for example, the set $\DS_0= \s{|S_3,S_4,S_5}$ on the path
$x_1=0,x_2=0$. The set of D-sequents listed before '$|$' is empty in $\DS_0$. This means that no D-sequents had been derived 
when \slv~ was called with $\pnt{q}=(x_1=0,x_2=0)$. On the exit of this invocation of \slv, D-sequents $S_3,S_4,S_5$ were derived.
We use ellipsis after symbol '$|$' for the calls of \slv~ that were not finished by the time $F$ was proved satisfiable.

Below, we use Figures~\ref{fig:tree} and~\ref{fig:derived_dseqs} to illustrate various aspects of the work of \slv.

\ti{Leaf nodes} correspond to subspaces where every variable is either assigned or proved redundant. For example,
the node on the path $(x_1=0,x_2=0,x_3=0,x_4=0)$ is a leaf because $x_1,x_2,x_3,x_4$ are assigned and
 $x_5$ is proved redundant.

\ti{Atomic D-sequents}. D-sequents $S_1,S_2,S_4,S_5$ are atomic. For example, the D-sequent $S_1$ is derived in
subspace $x_1=0,x_2=0,x_3=0,x_4=0$ due to  $x_5$ becoming  monotone. 
$S_1$ is equal to $(x_1=0,x_4=0) \rightarrow \s{s_5}$ because only assignments $x_1=0,x_4=0$ are responsible
for  the fact that  $x_5$ is monotone.

\ti{Branching in the presence of a conflict}. On the path $x_1=0,x_2=0$, clauses $C_3$ and $C_4$ turned into
unit clauses $x_3$ and $\overline{x}_3$ respectively. So no matter how first assignment to $x_3$ was made,
one of these two clauses would get falsified. \slv~made first assignment $x_3=0$ and falsified clause $C_3$. Since this
was the \ti{left branch} of $x_3$, \slv~proceeded further to branch on variable $x_4$.

\ti{Merging results of branches}. When branching on variable $x_4$, \slv~derived sets $\DS_0=\s{S_1}$ and $\DS_1=\s{S_2}$
where $S_1$ is equal to  $(x_1=0,x_4=0) \rightarrow \s{x_5}$ and $S_2$ is equal to  $(x_1=0,x_4=1) \rightarrow \s{x_5}$.
\slv~merged the results of branching by joining $S_1$ and $S_2$ at the branching variable $x_4$. The resulting D-sequent $S_3$ equal to
 $(x_1=0) \rightarrow \s{x_5}$ does not depend on $x_4$.

\ti{D-sequents for branching variables.} \slv~generated D-sequents for branching variables  $x_4$ and $x_3$.
Variable $x_4$ was monotone in subspace $x_1=0,x_2=0,x_3=0$ because the clauses $C_5$,$C_6$ containing the positive literal of $x_4$ were
not present in this subspace.  $C_5$ was satisfied by assignment $x_1=0$ while  $C_6$ contained variable $x_5$
whose redundancy was stated by D-sequent $S_3$ equal to  $(x_1=0) \rightarrow \s{x_5}$.
So the D-sequent $S_4$ equal to $(x_1=0) \rightarrow \s{x_4}$ was derived.

Variable $x_3$ was not monotone in subspace $\pnt{q}=(x_1=0,x_2=0)$ because, in this subspace, clauses $C_3$ and $C_4$ turned into unit clauses
$x_3$ and $\overline{x}_3$ respectively. So first, \slv~\ti{made} variable $x_3$ redundant by adding to $F$ clause $C_9= x_1 \vee x_2$ 
obtained by resolution of  $C_3$ and $C_4$ on $x_3$.  Note that $C_9$ is falsified in subspace \pnt{q}.  So the D-sequent $S_5$ equal
to $(x_1=0,x_2=0) \rightarrow \s{x_3}$ was generated.

\ti{Reduction of the size of  right branches.} In the left branch of node $x_2$, the set of D-sequents $\DS_0=\s{S_3,S_4,S_5}$ was
derived.  D-sequent $S_5$  equal to $(x_1=0,x_2=0) \rightarrow \s{x_3}$  is not symmetric in $x_2$ (i.e. depends on $x_2$).
On the other hand, $S_3$ and $S_3$ stating redundancy of $x_4$ and $x_5$ are symmetric in $x_2$.  So only D-sequent $S_5$ was
inactive in the right branch $x_2=1$ . So only variable $x_3$ reappeared in this branch while $x_4$,$x_5$ remain redundant.

\ti{Termination}. In subspace  $\pnt{q}=(x_1=0,x_2=1,x_3=1)$, every variable of $F$ was assigned or redundant
and no clause of $F$ was falsified by \pnt{q}. So \slv~terminated reporting that $F$ was satisfiable.
\end{example}
%
%
\subsection{Correctness of \slv}
The proof of correctness of \slv~can be performed 
by induction on the number of derived D-sequents.  Since such a proof is very similar to
the proof of correctness of   the quantifier elimination
algorithm we gave  in~\cite{tech_rep1}, we omit it here. Below we just list the facts on which
this proof of correctness is based.
\begin{itemize}
\item \slv~derives correct atomic D-sequents.
\item  D-sequents obtained by the join operation are correct.
\item \slv~correctly reports satisfiability when every clause is either satisfied or
proved redundant in the current subspace because D-sequents stating redundancy of variables are correct.
\item New clauses added to the current formula are obtained by resolution and so are correct.
So \slv~correctly reports unsatisfiability when an empty clause is derived.

\end{itemize}

\section{\slv~on Compositional Formulas}
\label{sec:comp_formulas}

In this section, we consider the performance of \slv~on compositional formulas. We will say that a satisfiability checking \tb{algorithm is compositional}
if its complexity is \ti{linear} in the number of subformulas forming a compositional formula.
 We prove that \slv~with lazy backtracking is compositional and  argue that   DPLL-based SAT-solvers are not.

We say that \tb{a formula} {\boldmath $F(X)$} \tb{is compositional} if it can be represented as $F_1(X_1) \wedge \ldots \wedge F_k(X_k)$
where $X_i \cap X_j = \emptyset, i \neq j$. The motivation for our interest in such formulas is as follows.
As we mentioned in Section~\ref{sec:comparison}, a practical formula $F$ typically can be represented 
as $F_1(X_1,Y_1) \wedge \ldots \wedge F_k(X_k,Y_k)$
where $X_i$ are internal variables of $F_i$ and $Y_i$ are communication variables.
 One can view compositional formulas as a degenerate case
where $|Y_i|=0,i=1,\ldots k$ and so $F_i$ do not talk to each other.
 Intuitively, an algorithm that does not scale well even when
$|Y_i|=0$  will not scale well when  $|Y_i| > 0$.

From now on, we \tb{narrow down the definition of compositional formulas} as follows. We will call formula 
 $F_1(X_1) \wedge \ldots \wedge F_k(X_k)$ compositional if  $X_i \cap X_j = \emptyset, i \neq j$ and all subformulas
$F_i$, $i=1,\ldots,k$ are equivalent modulo variable renaming/negation. That is $F_j$ can be obtained
from $F_i$ by renaming some variables of $F_i$ and then negating some variables of the result of   variable renaming.

\begin{proposition}
\label{prop:compos_form}
Let $F(X)=F_1(X_1) \wedge \ldots \wedge F_k(X_k)$ be a compositional formula. Let $T$ be the search tree built
by \slv~with lazy backtracking when checking the satisfiability of $F$.  The size of $T$
is linear in $k$ no matter how decision variables are chosen. (A variable $v \in X$ is a decision one
if no clause of $F$ that is unit in the current subspace contains $v$.)
\end{proposition}
\vspace{-5pt}
\begin{proof}
We will call a D-sequent \Dds{r}{\s{v}} \tb{limited to subformula} {\boldmath $F_i$} if $(\Va{r} \cup \s{v}) \subseteq \V{F_i}$.
The idea of the proof is to show that every D-sequent derived by \slv~is limited to a subformula $F_i$. 
Then the size of $T$ is limited by $|X| \cdot 2^d$ where $d = |\V{F_1}|=\ldots=|\V{F_k}|$.
Indeed, when \slv~flips the value of a variable $v$, only variables whose D-sequents depend on $v$ reappear
in the right branch of $v$. Since all D-sequents derived by \slv~are limited to  a subformula, 
the D-sequents depending on $v$ are  limited to subformula $F_i$ such that $v \in \V{F_i}$.
This means that the number of variables that reappear in the right branch is limited by $d$. So the number
of nodes of a right branch of $T$ cannot be larger than $2^d$. Hence  the size of $T$ cannot be larger
than  $|X| \cdot 2^d$ where $|X|$ is the maximum possible depth of $T$.

Now let us prove that every D-sequent derived by \slv~is indeed limited to a subformula $F_i$. Since subformulas
$F_i$,$F_j$, $i\neq j$ do not share variables, for any non-empty resolvent clause $C$, it is true that $\V{C} \subseteq  \V{F_i}$
for some $i$. Then any atomic D-sequent built for a monotone variable $v$ (see Subsection~\ref{subsec:atomic_dseqs}) is limited  
to the formula $F_i$ such that $v \in \V{F_i}$. \slv~builds an atomic D-sequent of another type   when a clause $C$ 
produced by resolution on branching variable $v$ is falsified in the current subspace. This D-sequent  has the form
\Dds{r}{\s{v}} where \pnt{r} is the shortest assignment falsifying  $C$.
Since $(\V{C} \cup \s{v}) \subseteq \V{F_i}$ where $F_i$ is the subformula containing $v$, such a D-sequent is limited to $F_i$.
Finally, a D-sequent obtained by joining D-sequents limited to $F_i$ is limited to $F_i$ $\square$
\end{proof}

Let \cl~be a DPLL-based algorithm with clause learning.  \cl~cannot solve compositional formulas $F_1 \wedge \ldots \wedge F_k$
in the time linear in $k$ for an arbitrary  choice of decision variables. Since every resolvent
clause can have only variables of one subformula $F_i$, the total number of clauses generated by \cl~is linear in $k$.
However, the time \cl~has to spend to derive one clause is also linear $k$. 
When a conflict occurs, \cl~backtracks to the  decision level that is \ti{relevant} to the conflict and  is the closest to the conflict level.
In the worst case, \cl~has to undo assignments of all $k$ subformulas.  So in the worst case, the complexity  of \cl~is quadratic in $k$.

Notice that the DP procedure is compositional because clauses of different subformulas cannot
be resolved with each other. However, as we mentioned in the introduction, this procedure is limited to one
global variable order in which variables are eliminated. This limitation is the main reason why the DP procedure is
outperformed by DPLL-based solvers. 
On the contrary, \slv~is a branching algorithm that can use different variable orders in different branches (and DPLL-based
SAT-solvers are a special case of \slv). So the machinery of D-sequents allows one to enjoy  the flexibility of branching
still preserving the compositionality of the algorithm.

\section{Skipping Right Branches}
\label{sec:skipping}
\setlength{\intextsep}{2pt}
\setlength{\textfloatsep}{2pt}
\begin{wrapfigure}{L}{2.3in}
\small
\vspace{-10pt}
\begin{tabbing}
aaaaa\=bb\=cc\=dd\= \kill
\Slv($F$,\pnt{q},\DS)\{\\
  \> ..... \\
\tb{\scriptsize{16}}  \> if ($\Sup{\DS}{asym} = \emptyset$)  return($F,\DS_0$);\\
\tb{\scriptsize{16.1}}\> if ($\mi{decision\_var}(\pnt{q_0},v,F)$) \\
\tb{\scriptsize{16.2}}\Tt if ($\mi{no\_new\_falsif\_clause}(\pnt{q_0},F)$)\{\\
\tb{\scriptsize{16.3}}\ttt   $S:= \mi{branch\_var\_Dseq}(F,v,\pnt{q})$; \\
\tb{\scriptsize{16.4}}\ttt   \DS:= $\mi{recomp\_Dseqs}(\DS_0,S,\pnt{q_0},F)$; \\
\tb{\scriptsize{16.5}}\ttt    return($F,\DS \cup \s{S}$); \}\\
\tb{\scriptsize{17}}  \> $\mi{recover\_vars\_clauses}(F,\Sup{\DS}{asym})$; \\
 \> .... \\

\tb{\scriptsize{21}}\> return($F,\DS$);\} \\
\end{tabbing} 
\vspace{-25pt}
\caption{modified \slv~procedure}
\label{fig:modification}
\end{wrapfigure}

In this section, we describe an optimization technique that can be used  for additional pruning the search tree built by \slv.
We will refer to this technique as SRB (Skipping Right Branches). The essence of SRB is that in some situations, \slv~
can use the D-sequents produced in the left branch of a variable $v$ to build D-sequents 
that do not depend on $v$ without exploration of the right branch of $v$.

This section is structured as follows. Subsection~\ref{subsec:modified_proc} gives pseudocode of the modification of \slv~with SRB.
Generation of D-sequents that do not depend on the current branching variable is explained in Subsection~\ref{subsec:generated_Dseqs}.
Some notions introduced in~\cite{tech_rep1} are recalled in Subsection~\ref{subsec:old_notions}. These notions
are used in Subsection~\ref{subsec:correctness_of_Dseqs} to prove that the D-sequents derived by the modified part of \slv~are correct.

\subsection{Modified \slv}
\label{subsec:modified_proc}
The modification   of \slv~ due to adding the SRB technique is shown in Figure~\ref{fig:modification} (lines 16.1-16.5).
SRB works as follows. Suppose that \slv~has backtracked from
the the left branch of $v$. Let \pnt{q} be the set of assignments made by \slv~before variable $v$.
 We will follow the assumption of Figure~\ref{fig:high_level_descr}
that the first value assigned to $v$ is  0. In such a case, \slv~of  Figure~\ref{fig:high_level_descr} just explores the
right branch $v=1$ (line 19). The essence of \slv~with SRB is that if a condition described below is satisfied,
the right branch is skipped. Instead, \slv~does the following. First, it builds a correct D-sequent of the
branching variable $v$ depending only on assignments to \pnt{q} (line 16.3). Then,  every D-sequent \Dds{r}{\s{w}} of $\DS_0$ 
where \pnt{r} contains assignment $(v=0)$ is recomputed (line 16.4).

Let \pnt{q_0} denote assignment \pnt{q} extended by $(v=0)$.
The condition under which the SRB technique is applicable is that 
no clause of $F$ having literal $v$ (i.e. the positive literal of variable $v$) is falsified by \pnt{q_0}. This means
that every clause with literal $v$ is either satisfied by \pnt{q} or has a variable that
is redundant in subspace \pnt{q_0}. This condition is checked on line 16.2.

The SRB technique is used in the modification of \slv~shown in Figure~\ref{fig:modification} only if 
$v$ is a decision variable (line 16.1). The reason is as follows.
Suppose that this is not the case, i.e. $v$ is in a unit clause $C$ of \cof{F}{q}.
In this case, in the left branch (respectively right branch), \slv~assigns $v$ the value that satisfies  $C$
(respectively falsifies $C$). But since \slv~immediately backtracks if a new clause gets falsified in the right
branch, pruning the left branch in this case does not save any work.

\subsection{D-sequents generated by  modified \slv}
\label{subsec:generated_Dseqs}
Let $F(X)$ be a CNF formula.
Let \pnt{q} be a partial assignment to variables of $X$.  Let $v$ be a variable of $X \setminus \Va{q}$.
Let \pnt{q_0} denote the assignment $\pnt{q} \cup \s{(v=0)}$. Let $\DS_0$ be a set of D-sequents
active in the subspace specified by \pnt{q_0}. Let every clause of $F$ that has literal
 $v$ is either satisfied by \pnt{q} or has a variable whose redundancy is stated by a D-sequent
of $\DS_0$.

The procedure \ti{branch\_var\_Dseq} of Figure~\ref{fig:modification} generates D-sequent \Dds{r}{\s{v}} such
that for every clause $C$ containing literal $v$
\begin{itemize} 
     \item $C$ is satisfied by \pnt{r} or
\item $C$ contains a variable $w$ whose redundancy is stated by a  D-sequent \Dds{s}{\s{w}} of $\DS_0$
      and $\pnt{s} \leq (\pnt{r} \cup \s{(v=0)}$.
\end{itemize}

The procedure \ti{recomp\_Dseqs} of Figure~\ref{fig:modification} works as follows.
For every D-sequent \Dds{e}{\s{w}} of $\DS_0$ such that \pnt{e} contains assignment $(v=0)$,
\ti{recomp\_Dseqs} generates a  D-sequent \Dds{e'}{\s{w}}. The assignment \pnt{e'} is obtained from \pnt{e} by replacing assignment
       $(v=0)$ with  the assignments of \pnt{r} of the D-sequent \Dds{r}{\s{v}} generated for the branching variable $v$.

\subsection{Recalling some notions}
\label{subsec:old_notions}
In this subsection, we recall some notions introduced in~\cite{tech_rep1} that are used in the proofs of Subsection~\ref{subsec:correctness_of_Dseqs}.
Let $F(X)$ be a CNF formula. We will refer to a complete assignment to variables of $X$ as a \tb{point}.
A point \pnt{p} is called a \pnt{Z}\tb{-boundary point} of $F$ if 
\begin{itemize}
\item \pnt{p} falsifies $F$
\item every clause of $F$ falsified by \pnt{p} contains a variable of $Z$
\item $Z$ is minimal i.e. no proper subset of $Z$ satisfies the property above
\end{itemize} 

A $Z$-boundary point \pnt{p} is called \pnt{Y}\tb{-removable} in $F$ where $Z \subseteq Y \subseteq X$ if \pnt{p}
cannot be turned into an assignment satisfying $F$ by changing values of variables of $Y$. If a $Z$-boundary point
is $Y$-removable, then one can produce a clause $C$ that is 
a) falsified by \pnt{p}; b) implied by $F$ and c) does not have any variables of $Z$.
After adding $C$ to $F$, \pnt{p} is not a $Z$-boundary point anymore, hence the name removable.

We will
call a $Y$-removable  point \tb{just removable} if $Y = X$. It is not hard to see, that every $Z$-boundary point
of a satisfiable (respectively unsatisfiable) formula $F$ is unremovable (respectively removable).

\begin{proposition}
\label{prop:pnts_vars}
Let $F(X)$ be a CNF formula and \pnt{q} be a partial assignment to variables of $X$.
A set of variables $Z$ is not redundant in \prob{X}{F} in subspace \pnt{q}, if and only if
there is a $Z$-boundary point of \cof{F}{q} that is removable in $F$.
\end{proposition}
The proof of this proposition is given in ~\cite{tech_rep1}.

\subsection{Correctness of  D-sequents generated by modified \slv}
\label{subsec:correctness_of_Dseqs}
%
%

\begin{proposition}
\label{prop:SRB1}
The D-sequent generated by procedure \ti{branch\_var\_Dseq} of Figure~\ref{fig:modification} described in
 Subsection~\ref{subsec:generated_Dseqs} is correct.
\end{proposition}
\begin{proof}  Assume the contrary i.e.  D-sequent \Dds{r}{\s{v}} does not hold. From Proposition~\ref{prop:pnts_vars}.
it follows that there is  a \s{v}-boundary point \pnt{p} such that $\pnt{r} \leq \pnt{p}$. This also means that $F$ is unsatisfiable. 
Indeed, if $F$ is a satisfiable, then  every variable of $F$ is already redundant and so 
any D-sequent  holds \Dds{r}{\s{v}}.

Let us assume that $v$ is equal to 0
in \pnt{p}. If $v$ equals 1 in \pnt{p}, one can always flip the value of $v$ obtaining  the point that is either a \s{v}-boundary
point or a satisfying assignment (Lemma 1 of \cite{tech_rep1}). Since the assumption we made implies that $F$ is unsatisfiable,
flipping the value of $v$ produces a \s{v}-boundary point.

Let $G$ be the set clauses falsified by point \pnt{p}. Let $C$ be a clause of $G$. Since \pnt{p} is a \s{v}-boundary point,
then $C$ contains literal $v$.
Note that under the assumption of the proposition to be proved,
if a clause of $F$ with literal $v$ is not satisfied by \pnt{r}, this clause has to contain a redundant variable
$w$ such that D-sequent \Dds{s}{\s{w}} of $\DS_0$ holds and $\pnt{s} \leq (\pnt{r} \cup \s{(v=0)}$. Let $Z$ be a minimal set of variables of $X$ 
that are present in clauses of $G$ and whose redundancy is stated by D-sequents of $\DS_0$. Then \pnt{p} is
a $Z$-boundary point of $F$. Since $F$ is unsatisfiable, this point is removable. Then from Proposition~\ref{prop:pnts_vars} 
it follows that the variables of $Z$ are not in redundant
in \cof{F}{r'} where $\pnt{r'} = \pnt{r} \cup \s{(v=0)}$. Contradiction.
\end{proof}

\begin{proposition}
\label{prop:SRB2}
The D-sequents generated by the \ti{recomp\_Dseqs} procedure of Figure~\ref{fig:modification} 
described in Subsection~\ref{subsec:generated_Dseqs} are correct.
\end{proposition}
\begin{proof} Assume the contrary i.e. the D-sequent \Dds{e'}{\s{w}} obtained
from a D-sequent \Dds{e}{\s{w}} of $\DS_0$ does not hold.  This means that there is a \s{w}-boundary
point \pnt{p} such that $\pnt{e'} \leq \pnt{p}$.  It also means that $F$ is unsatisfiable.
Let us consider the following two cases.
\begin{itemize}
\item Variable $v$ is assigned 0 in \pnt{p}. Then there is a removable \s{w}-boundary point in subspace \pnt{e} and so
the D-sequent \Dds{e}{\s{w}} of $\DS_0$ does not hold. Contradiction.
\item Variable $v$ is assigned 1 in \pnt{p}. Let \pnt{p'} be the point obtained from \pnt{p} by flipping the value of $v$.
Let $G$ and $G'$ be the clauses of $F$ falsified by \pnt{p} and \pnt{p'} respectively. Denote by $G''$ the set of clauses $G' \setminus G$.
This set consists only of clauses having literal $v$ because these are the only new clauses that may get falsified after
flipping the value of $v$ from 1 to 0. Then using reasoning similar to that of Proposition~\ref{prop:SRB1}, one concludes that every clause
of $G''$ contains a variable $u$ such that  D-sequent \Dds{s}{\s{u}} of $\DS_0$ holds and $\pnt{s} \leq (\pnt{r}  \cup \s{(v=0)})$
where \pnt{r} is the assignment of the D-sequent \Dds{r}{\s{v}} generated for the branching variable $v$.
Let $Z$ be a minimal set of such variables. Then any clause of $G'$ either contains variable $w$ or a variable of $Z$.
Hence \pnt{p'} is a ($Z \cup \s{w}$)-boundary point. Since our assumption implies unsatisfiability of $F$, 
this boundary point is removable.
Let \pnt{g} be equal to $\pnt{e'} \cup \s{(v=0)}$.
Note that since $\pnt{r} \leq \pnt{e'}$ and \pnt{g} contains assignment $(v=0)$ every variable of
$Z$ is  redundant in \cof{F}{g}. On the one hand, since $\pnt{e} \leq \pnt{g}$, variable $w$ is redundant
in \cof{F}{g} as well. So the variables of $Z \cup \s{w}$ are redundant in \cof{F}{g}.
On the other hand, $\pnt{g} \leq \pnt{p'}$ and so \cof{F}{g} contains a  $(Z \cup \s{w})$-boundary point \pnt{p'}
that is removable in $F$. From Proposition~\ref{prop:pnts_vars}, it follows that variables of $(Z \cup \s{w})$
are not redundant in \cof{F}{g}.
  Contradiction.
\end{itemize}
\end{proof}

\section{Experiments}
\label{sec:experiments}

In this section, we compare \slv~with some well-known SAT-solvers on two sets of compositional and non-compositional formulas. 
In experiments, we used   the optimization technique described in  Section~\ref{sec:skipping}. Although, using this technique was not
crucial for making our points, it allowed to improve the runtimes of \slv.

Obviously, this comparison 
by no way is  comprehensive. Our objective here is as follows. In Subsection~\ref{subsec:eager_lazy}, we argued
that DPLL-based SAT-solvers is a special case \slv~when it uses eager backtracking. One may think
that due to great success of modern SAT-solvers, this version of~\slv~is simply always the best. In this section,
we show that is not the case. We give an example of meaningful formulas where the opposite strategy of lazy
backtracking works much better. This result confirms the theoretical prediction of Section~\ref{sec:comp_formulas}.

\begin{wraptable}{l}{2.7in}
\small
\caption{\ti{Solving compositional formulas}}
\scriptsize
\begin{center}
\begin{tabular}{|c|c|c|c|c|c|c|} \hline
 \#copi- & \#vars           & \#clauses         & minisat & rsat   & picosat   & ds-qsat   \\ 
 es $\times 10^3$  &  $\times 10^3$   & $\times 10^3$    &  (s.)    &  (s.)   &   (s.)    & (s.) \\ \hline
 5    &   80             & 170              &  9.1     &  5.1    &  4.2     &  0.7     \\ \hline
 10   &   160           & 340              &  110     &  28      &  20        & 1.6    \\ \hline
 20   &   320           & 680             &  917     &  143      &  80        & 3.3   \\ \hline
 40   &   640           & 1,360          &  $>$ 1hour     &    621    &  305       & 7.2   \\ \hline
 80   &   1,280           & 2,720          &  $>$ 1hour     &   2,767     &   1,048       &  15  \\ \hline
\end{tabular}                
\end{center}
\label{tbl:compos_form}
\end{wraptable}

The results of experiments with the  first set of formulas are shown in Tables~\ref{tbl:compos_form}, ~\ref{tbl:compos_form_stat1}
and ~\ref{tbl:compos_form_stat2}
This set consists of compositional formulas $F_1(X_1) \wedge \ldots \wedge F_k(X_k)$ where $X_i \cap X_j = \emptyset$. Every subformula $F_i$
is obtained by renaming/negating variables of the same satisfiable CNF formula describing a 2-bit multiplier.
 Since every subformula $F_i$ is satisfiable, then formula $F_1 \wedge \ldots \wedge F_k$ is satisfiable too for any value of $k$.

\begin{wraptable}{l}{2.4in}
\small
\caption{\ti{Statistics of Picosat and \slv~on compositional formulas}}
\scriptsize
\begin{center}
\begin{tabular}{|c|c|c|c|c|c|c|} \hline
 \#co- & \multicolumn{3}{|c|}{picosat} & \multicolumn{3}{|c|}{ds-sat} \\ 
\cline{2-4}\cline{5-7}   
 pies & \#cnfl. & \#dec. & \#impl. & \#cnfl. & \#dec. & \#impl.  \\ 
   $\times 10^3$   & $\times 10^3 $   & $\times 10^6$  & $\times 10^6$ &  $\times 10^3 $ &$\times 10^3$ & $\times 10^3$   \\ \hline
 5                 &   0.8   &  3   & 12                             & 0.8 & 25 & 62  \\ \hline
 10                &  1.5    & 10     & 39                             & 1.5 & 50 & 122      \\ \hline
 20                & 3.0     &  40    & 144                            & 3.2 & 100 & 245  \\ \hline
 40                &  5.5    & 138    & 489                            & 6.5  & 199 & 493      \\ \hline
 80                & 9.7     & 443    &1,533                           & 13.0 & 399 & 985   \\ \hline
\end{tabular}                
\end{center}
\label{tbl:compos_form_stat1}
\end{wraptable}

In the DIMACS format that we used in experiments, a variable's name
is a number.  In the formulas of  Table~\ref{tbl:compos_form}, the variables were named so 
that the DIMACS names  of variables of different subformulas $F_i$ interleaved. The objective of  negating variables was to make sure
that if an assignment \pnt{s} to the variables of $X_i$ satisfies $F_i$, the same assignment of the corresponding variables of $X_j$ is unlikely
to satisfy $F_j$.

In Table~\ref{tbl:compos_form},
we compare \slv~ with Minisat (version 2.0), RSat (version 2.01) and Picosat (version 913) on compositional formulas.
These formulas are different only in the value of $k$. The first three columns of this table show the value of  $k$, 
the number of variables and clauses in thousands. The last four columns show the time taken by Minisat, RSat,Picosat and \slv~ to solve
these formulas (in seconds).  \slv~significantly outperforms these three SAT-solvers. As predicted by Proposition~\ref{prop:compos_form},
\slv~shows linear complexity. On the other hand, the complexity of  each of the three SAT-solvers is proportional to $m \cdot k^2$
where $m$ is a constant.

Table~\ref{tbl:compos_form_stat1} provides some statistics of the performance of Picosat and \slv~on the formulas of Table~\ref{tbl:compos_form}.
The second, third and fourth columns give the number of conflicts (in thousands),  number of decision and implied assignments (in millions)
for Picosat. In the following three columns, the number of conflict nodes of the search tree,
number of decision and implied assignments (in thousands) are given for \slv. The results of Table~\ref{tbl:compos_form_stat1}
show that the numbers of conflicts of Picosat and those of conflict nodes of \slv~are comparable. Besides, for both programs 
the dependence of these numbers  on $k$ is linear. However, the numbers of decision and implied assignments made by Picosat and \slv~
differ by three orders of magnitude.
 Most importantly, the number of assignments made by \slv~
(both decision and implied) 
grows linearly  with $k$. On the other hand,  the dependence of the  number of assignments 
made by Picosat  on $k$ is closer to quadratic for both decision and implied assignments.

Table~\ref{tbl:compos_form_stat2} provides some
 additional statistics characterizing the performance of  \slv~on  the formulas of Table~\ref{tbl:compos_form}.
The second column
specifies the maximum number of conflict variables that appeared on a path of the search tree. A variable $v$ is a conflict one
if after making an assignment to $v$ a new clause of $F$ gets falsified. This column  shows that \slv~kept  branching
even after  thousands of conflicts occurred on the current path. 

\begin{wraptable}{l}{1.5in}
\small
\caption{\ti{More statistics of \slv~for compositional formulas}}
\scriptsize
\begin{center}
\begin{tabular}{|c|c|c|c|c|} \hline
  \#vars                   & max   &  \#assgn.  & max    \\ 
  $\times 10^3$     & confl    & vars in & right \\ 
                    & vars   &   sol.   (\%) &branch         \\ \hline
   80              & 505        &    2 & 14    \\ \hline
   160             & 1,027     &   0.1 & 14        \\ \hline
   320             & 1,956     &   1  & 14    \\ \hline
   640             & 3,795      &   1 & 14      \\ \hline
   1,280           & 7,351      &  0.2  & 14    \\ \hline
\end{tabular}                 
\end{center}
\label{tbl:compos_form_stat2}
\end{wraptable}

\slv~reports that a formula is satisfiable when the current 
assignment \pnt{q} does not falsify a clause of $F$ and every variable of $F$ that is  not assigned in \pnt{q} is proved redundant. The third
column of Table~\ref{tbl:compos_form_stat2} 
gives the value of $|\Va{q}| / |\V{F}|$ (in percent) at the time \slv~proved satisfiability. Informally, this value
shows that \slv~established satisfiability of $F$ knowing only a very small fragment of a satisfying assignment.
The last column of Table~\ref{tbl:compos_form_stat2} shows the maximum number of non-redundant unassigned variables that appeared
in a right branch of the search tree. The number of variables in subformulas $F_i$ we used in experiments
was equal to 16. As we showed in Proposition~\ref{prop:compos_form}, in the search tree built by \slv~for a compositional formula,
the number of free variables that may appear in a right branch is bounded by $|V(F_i)|$ i.e by 16. Our experiments confirmed that prediction.
The fact that the size of  right branches is so small  means that when solving a formula $F$ of Table~\ref{tbl:compos_form}, \slv~dealt
 only with very small fragments of $F$.

\begin{wraptable}{l}{2.7in}
\small
\caption{\ti{Solving non-compositional formulas}}
\scriptsize
\begin{center}
\begin{tabular}{|c|c|c|c|c|c|c|} \hline
 \#sub-        & \#vars & minisat    & rsat    & picosat   & ds-qsat & ds-qsat*   \\ 
 form.         &    $\times 10^3$    &(s.)        &  (s.)      &   (s.)        &   (s.)      &  (s.)      \\
 $\times 10^3$ &       &            &     &      &     &  \\ \hline
 5             & 75  & 5.0    & 3.2     &  3.6      &  5.1     &  0.4 \\ \hline
 10            & 150 & 34    &  21    &   15    &   13    &  1.0  \\ \hline
 20            & 300 & 548     &  79    &  57     &  30     & 2.3   \\ \hline
 40            & 600 & $>$1hour  & 430   &  231     &  57     & 5.9   \\ \hline
 80            & 1,200 & $>$1hour  &  1,869  &  859     & 127     &  19  \\ \hline
\end{tabular}                
\label{tbl:non_compos_form}
\end{center}

\end{wraptable}

Generally speaking, the problems with compositional formulas can be easily fixed by solving independent subformulas separately.
Such subformulas can be found in linear time by looking for strongly connected components of a graph relating clauses that share a variable.
To eliminate such a possibility we conducted the second experiment.
In this experiment, we compared \slv~ and the three SAT-solvers above on \ti{non-compositional formulas}. 
Those formulas were obtained
from the same subformulas $F_i$ obtained from a CNF formula  specifying a 2-bit multiplier by renaming/negating variables. However, now,
renaming was done in such a way that every pair of subformulas $F_i$,$F_{i+1}$, $i=1,\ldots,k-1$ shared exactly one variable.
So now formulas $F=F_1(X_1) \wedge \ldots \wedge F_k(X_k)$ we used in experiments did not have any independent subformulas. 
Table~\ref{tbl:non_compos_form} shows the results of the second experiment (all formulas are still satisfiable).
 The first two columns of Table~\ref{tbl:non_compos_form} specify
the value of $k$ (i.e. the number of subformulas $F_i$) and the number of variables of $F$ in thousands. The next four columns
give the runtimes of Minisat, RSat, Picosat and \slv~in seconds. These runtimes show that \slv~still outperforms these three SAT-solvers 
and scales better. 

The last column of Table~\ref{tbl:non_compos_form} illustrates the ability of D-sequents to take into account formula structure.
In this column, we give the runtimes of \slv~when it first branched on communication variables (i.e. ones shared by subformulas $F_i$).
 So in this case, \slv~had information about formula structure. 
 The results  show that the knowledge of communication variables considerably improved the performance of \slv.

\setlength{\intextsep}{2pt}
\setlength{\textfloatsep}{2pt}
\begin{wraptable}{l}{2.3in}
\small
\caption{\ti{Statistics of \slv~for non-compositional formulas}}
\scriptsize
\vspace{-10pt}
\begin{center}
\begin{tabular}{|c|c|c|c|c|c|} \hline
  \#vars           & max      &  \#assgn.  & max     & max  & max       \\ 
  $\times 10^3$    & confl    & vars in    & right   & right & right     \\ 
                   & vars     &   sol.     & branch  & branch & branch*   \\
                   &          &   (\%)     &         & (\%)  &          \\ \hline

   75              &  461     &   4        &  338    &  0.5 & 11\\ \hline
   150             &  903     &   1        &  475    & 0.3   & 11     \\ \hline
   300             &  1,765   &   1        &  571    & 0.2   & 11\\ \hline
   600             &  3,512   &   2        &  773    & 0.1  &  11\\ \hline
   1,200           &  7,029   &   1        &  880    & 0.1   &  11\\ \hline
\end{tabular}                 
\label{tbl:non_compos_form_stat}
\end{center}

\end{wraptable}

Table~\ref{tbl:non_compos_form_stat} gives some statistics describing the performance of  \slv~on  the formulas of Table~\ref{tbl:non_compos_form}.
The second and third columns of Table~\ref{tbl:non_compos_form_stat} are similar to the corresponding columns of Table~\ref{tbl:compos_form_stat2}.
A lot of conflicts occurred on a path of the search tree built by \slv~ and by the time \slv~reported satisfiability, only a small fragment of a satisfying
assignment was known. The fourth column shows the maximum size of a right branch 
 of the search tree built by \slv~(in terms of the number of non-redundant variables). 
The next column gives the ratio of the maximum size of a right branch and the total number
of variables (in percent). Notice, that now the maximum size of right branches is much larger than in the case of compositional formulas.
Nevertheless, again, \slv~dealt only with very small fragments of the formula. The last column gives  the maximum size of right branches
when \slv~first branched on communication variables.  Such structure-aware branching allowed \slv~to dramatically reduce 
the size of right branches, which explains why \slv~had much faster runtimes in this case
 (shown in the last column of Table~\ref{tbl:non_compos_form}).

Although \slv~performed well on the formulas we used in experiments, lazy backtracking  is  too extreme to be successful
on a more general set of benchmarks. Let $F$ be a formula to be checked for satisfiability. Let a clause $C$ of $F$ be falsified
by the current assignment and $Z$ be the set of unassigned variables.  At this point, any variable $v \in Z$ is redundant due to
$C$ being falsified.
Lazy backtracking essentially assumes that by keeping branching in the presence of the conflict one will find a better explanation of redundancy of $v$.
 A less drastic approach is as follows. Once clause $C$ gets falsified, 
a D-sequent \Dds{r}{Z'} is derived where \pnt{r} is the shortest assignment falsifying $C$ and $Z'$ consists of some variables
of $Z$ that are related to clause $C$. For example, if $C$ is in a subformula $G$ of $F$ specifying a design block
it may make sense to form $Z'$ of the unassigned variables of $G$.

\section{Background}
\label{sec:background}
In 1960, Davis and Putnam  introduced a QSAT-solver that is now called the DP procedure~\cite{dp}.
Since it performed poorly even on small formulas,  a new algorithm called the DPLL procedure
was introduced in 1962 ~\cite{dpll}.  Two major changes were made in the DPLL procedure in comparison to the DP procedure.
First, the DPLL procedure employed branching and could use different variable order in different branches.
Second, it changed the semantics of variable elimination that the DP procedure was based on.
Instead, the semantics of  elimination of unsatisfiable assignments was introduced.
The DPLL procedure backtracks as soon as it finds out that the current partial assignment cannot be
extended to a satisfying assignment. Such eager backtracking is a characteristic feature of SAT-solvers i.e. algorithms
proving satisfiability by producing a satisfying assignment.

The first change has been undoubtedly a great step forward. The DP procedure  eliminates variables
in one particular global order which makes this procedure very inefficient. The second change
however has its pros and cons. On the one hand, DPLL has a very simple and natural semantics, which facilitated
the great progress in SAT-solving  seen in the last two decades~\cite{grasp,sato,chaff,berkmin,minisat,rsat,picosat}.
On the other hand, as we argued before, the necessity to generate a satisfying assignment to prove satisfiability deprived
DPLL-based SAT-solvers of powerful transformations preserving equisatisfiability rather than functional equivalence.

Generally speaking, transformations preserving equisatisfiability are routinely used by modern algorithms, but their usage is limited
one way or another.
For example such transformations  are employed in preprocessing where some  variables are resolved out~\cite{prepr} or redundant clauses are
removed~\cite{tacas_blocked_clauses}. However, such transformations have a limited scope: they are used just to simplify the original
formula  that 
is then passed to a DPLL-based SAT-solver. Second, transformations preserving equisatisfiability are ubiquitous in
algorithms on  circuit formulas e.g. in ATPG algorithms~\cite{atpg}. Such  algorithms often exploit the fact that
a gate becomes unobservable under some partial assignment \pnt{r}. In terms of variable redundancy, this  means that the variable $v$
specifying the output of this gate is redundant in subspace \pnt{r}. Importantly, this redundancy is defined with respect to
a formula where assignments to non-output variables do not matter. Such variables can be viewed as existentially quantified
and the discarding of  clauses containing redundant non-output variables does not preserve functional equivalence.
However, such transformations are restricted
only to formulas generated off circuits and do not form a complete calculus. Typically, these transformations 
are used in the form of heuristics.

The machinery of D-sequents was introduced in ~\cite{tech_rep1,tech_rep2}. In turn, the notion of D-sequents
and  join operation were inspired by the relation between variable redundancy and boundary point elimination~\cite{sat09,haifa10}.
In~\cite{tech_rep1,fmcad12}, we formulated a method of quantifier elimination called DDS (Derivation of D-Sequents).
Since QSAT is a special case of the quantifier elimination problem, DDS can be used to check satisfiability. 
However, since DDS employs eager backtracking,  such an algorithm
is a SAT-solver rather than  a QSAT-solver. In particular, as we showed in \cite{tech_rep1}, the complexity of DDS
on compositional formulas is quadratic in the number $k$ of subformulas, while the complexity of \slv~is linear in $k$. 

\section{Conclusion}
\label{sec:conclusion}
The results of this paper lead to the following three conclusions.

\vspace{5pt}
\noindent 1) DPLL-based procedures have scalability issues. These issues can be observed 
even on compositional formulas i.e. on formulas with a very simple structure.
Arguably, the root of the problem, is that DPLL-procedures are designed to prove
satisfiability by producing a satisfying assignment. This deprives such procedures
from using  powerful transformations that preserve equisatisfiability rather than functional equivalence.

\vspace{5pt}
\noindent 2) D-sequents are an effective  tool for building scalable algorithms. In particular,
the algorithm \slv~we describe in the paper scales well on compositional formulas.
The reason for such scalability is that \slv~scarifies the ability to generate satisfying
assignments to tap into the power of  transformations preserving only equisatisfiability.
The essence of transformations used by \slv~is to discard large portions of the formula
that are proved to be redundant in the current subspace. The results of experiments
with \slv~on two simple classes of compositional and non-compositional formulas show
the big promise of algorithms based on D-sequents.

\vspace{5pt}
\noindent 3) In this paper, we have only touched the  tip of the iceberg. A great deal of issues needs to be resolved
to make  QSAT-solving by D-sequents practical. 

\section{Acknowledgment}
This work was funded  in part by NSF grant CCF-1117184.
\bibliographystyle{plain}
\bibliography{short_sat,local}
\end{document}